\documentclass[english]{article}
\usepackage{geometry}
\usepackage{booktabs}
\usepackage{float}
\usepackage{mathrsfs}
\usepackage{amsmath}
\usepackage{amssymb}
\usepackage{graphicx}
\usepackage{amsfonts}
\usepackage{bm}
\usepackage{subfigure}
\usepackage{amsthm}
\usepackage{epsfig}
\makeatletter

\def\levy{L\'{e}vy}

\@ifundefined{definecolor}{\@ifundefined{definecolor}
 {\@ifundefined{definecolor}
 {\usepackage{color}}{}
}{}
}{}

\usepackage{subfig}\usepackage[all]{xy}

\newcommand{\bbR}{\mathbb R}

\newcommand{\testfunction}{\ensuremath{\mathcal{B}_b}}
\newtheorem{theorem}{Theorem}[section]
\newtheorem{lem}{Lemma}[section]
\newtheorem{rem}{Remark}[section]
\newtheorem{prop}{Proposition}[section]

\newcounter{hypA}
\newenvironment{hypA}{\refstepcounter{hypA}\begin{itemize}
  \item[({\bf A\arabic{hypA}})]}{\end{itemize}}
\newcounter{hypB}

\newcommand{\Exp}{\mathbb{E}}

\usepackage{babel}\date{}

\usepackage{babel}

\makeatother

\usepackage{babel}

\begin{document}

\begin{center}

{\Large \textbf{Biased Online Parameter Inference for State-Space Models}}

\vspace{0.5cm}

BY YAN ZHOU \& AJAY JASRA 

{\footnotesize Department of Statistics \& Applied Probability,
National University of Singapore, Singapore, 117546, SG.}\\
{\footnotesize E-Mail:\,}\texttt{\emph{\footnotesize stazhou@nus.edu.sg, staja@nus.edu.sg}}\\
\end{center}

\begin{abstract}
We consider Bayesian online static parameter estimation for state-space models. This is a very important problem, but is very computationally challenging
as the state-of-the art methods that are exact, often have a computational cost that grows with the time parameter; perhaps the most successful algorithm
is that of SM$\textrm{C}^{2}$ \cite{chopin}. We present a version of the SM$\textrm{C}^{2}$ algorithm which has computational cost that does not grow with the time parameter.
In addition, under assumptions, the algorithm is shown to provide consistent estimates of expectations w.r.t.~the posterior. However, the cost to achieve this consistency can be exponential
in the dimension of the parameter space; if this exponential cost is avoided, typically the algorithm is biased.
The bias is investigated from a theoretical perspective and, under assumptions, we find that the bias does not accumulate as the time parameter grows. 
The algorithm is implemented on several Bayesian statistical models.
\end{abstract}

\section{Introduction}

We consider a state-space models, that is, a pair of discrete-time stochastic
processes, $\left\{  X_{n}\right\}  _{n\geq 0}$ and $\left\{
Y_{n}\right\} _{n\geq 1}$. The hidden process
$\left\{  X_{n}\right\}_{n\geq 0} $ is a\ Markov chain and the joint density for $n$ observations of the hidden process and the observations is
$$
\nu_{\theta}(x_0)\prod_{k=1}^n f_{\theta}(x_k|x_{k-1}) g_{\theta}(y_k|x_k)
$$
where $x_k\in\mathsf{X}$, $k\geq 0$,  and $y_k\in\mathsf{Y}$, $k\geq 1$, $\theta\in\Theta\subseteq\mathbb{R}^d$
is a static parameter with prior $\pi_{\theta}$. In particular, it is of interest to infer the posterior on $(\theta,x_{0:n})$, conditional upon $(y_1,\dots,y_n)$
 ($x_{0:n}=(x_0,\dots,x_n)$) as the time parameter $n$ grows. This problem is of interest in a wide variety of applications,
including econometrics, finance and engineering; see for instance \cite{cappe}.

In general, even if $\theta$ is fixed, the posterior cannot be computed exactly and one often has to resort to numerical methods, for example by using sequential Monte Carlo (SMC) (see e.g.~\cite{doucet}). SMC  makes use of  a collection of proposal densities and sequentially simulates from these $N>1$ samples, termed particles. In most scenarios
it is not possible to use the distribution of interest as a proposal. Therefore, one must correct for the discrepancy between proposal
and target  via importance weights. In the majority of cases of practical interest, the variance of these importance weights increases with algorithmic time. This can, to some extent, be dealt with via a  resampling procedure consisting of sampling with replacement from the current weighted samples and resetting them to $1/N$.
However, as is well known in the literature, due to the path degeneracy problem for particle filters, when $\theta$ is a random variable SMC methods do not always work well; see the review of \cite{kantas} for details.
This has lead to a wide variety of techniques being developed, including \cite{chopin,dan,fearnhead,gilks,polson}; see \cite{kantas} for a full review.

The method which one might consider to be the state-of-the-art for the  Bayesian online static parameter estimation for state-space models, is that in \cite{chopin}. This approach
combines the methods of SMC samplers \cite{delm:06} and particle Markov chain Monte Carlo (PMCMC) \cite{Andrieu_2010}. The method provides consistent estimates (that is, as the number of samples grows) of expectations w.r.t.~the 
posterior on $(\theta,x_{0:n})$, as the time parameter grows. The method has been shown to work well in practice, but has one major issue; the computational cost of the application of PMCMC kernels grows as the time parameter grows;
whilst the amount of times that the application of the kernel may stablize, one will still need to apply the kernel during the algorithm (although a variety of tricks can be used to reduce the cost; see \cite{jacob}).
We present a version of the SM$\textrm{C}^{2}$ algorithm which has computational cost that does not grow with the time parameter.
In addition, under assumptions, the algorithm is shown to provide consistent estimates of expectations w.r.t.~the posterior. However, the cost to achieve this consistency can be exponential
in the dimension of the parameter space; if this exponential cost is avoided, typically the algorithm is biased.
The bias is investigated from a theoretical perspective and, under assumptions, we find that the bias does not accumulate as the time parameter grows. 

Our approach is based upon breaking up the observations into blocks of $T$ observations; an approach used, differently, in other articles such as \cite{chan,dan,gilks}.
The idea is simply to develop a `perfect' SM$\textrm{C}^{2}$ algorithm which has computational cost that does not grow with the time parameter; in practice, one
cannot implement this algorithm and so one must approximate it. Our approach simply uses the approximation of an appropriate target from the previous block.
It is remarked that if the posterior exhibits concentration/Bernstein-Von Mises properties as the time parameter grows, then alternative schemes are relevant, which could be prefered to the ideas here; however, such properties
do not always hold. See for instance the work of \cite{douc} in the context of MLE for relatively weak conditions; the properties of the MLE can impact on the concentration/Bernstein-Von Mises properties of the posterior (see e.g.~\cite{borwan}).

In Section \ref{sec:algo}, SM$\textrm{C}^{2}$ is reviewed and our algorithm is presented.
In Section \ref{sec:theory} our theoretical results are presented, with proofs housed in the appendix.
In Section \ref{sec:simos} our simulation results are presented. In Section \ref{sec:summary} the article is concluded
and several possible extensions are discussed.

\section{Algorithm}\label{sec:algo}

\subsection{Notations}

Let $(E,\mathcal{E})$ be a measurable space.
The notation $\testfunction(E)$ denotes the class of bounded and measurable real-valued functions. 
$\mathcal{C}_b(E)$ dentes the  continuous, bounded measurable real-valued functions on $E$.
The supremum norm is written as $\|f\|_{\infty} = \sup_{x\in E}|f(x)|$. $\mathscr{P}(E)$ is the set of probability measures on $(E,\mathcal{E})$. We will consider non-negative operators $K : E \times \mathcal{E} \rightarrow \bbR_+$ such that for each $x \in E$ the mapping $A \mapsto K(x, A)$ is a finite non-negative measure on $\mathcal{E}$ and for each $A \in \mathcal{E}$ the function $x \mapsto K(x, A)$ is measurable; the kernel $K$ is Markovian if $K(x, dy)$ is a probability measure for every $x \in E$.
For a finite measure $\mu$ on $(E,\mathcal{E})$, real-valued and measurable $f:E\rightarrow\mathbb{R}$
\begin{equation*}
    \mu K  : A \mapsto \int K(x, A) \, \mu(dx)\ ;\quad 
    K f :  x \mapsto \int f(y) \, K(x, dy).
\end{equation*}
We also write $\mu(f) = \int f(x) \mu(dx)$. $\|\cdot\|_{\textrm{tv}}$ denotes the total variation distance.


\subsection{SM$\textrm{C}^2$}\label{sec:smc2}

SMC$^2$ uses SMC to sample the following sequence of targets, for a fixed $N_x\geq 1$
$$
\pi_n(\xi_n) \propto \pi_{\theta}(\theta)\Psi_{\theta}(x_{0:n}^{1:N_x},a_{1:n}^{1:N_x})
\Big(\prod_{k=1}^n \frac{1}{N_x}\sum_{i=1}^{N_x} g_{\theta}(y_k|x_k^i)\Big)
$$
with $\xi_n=(\theta,x_{0:n}^{1:N_x},a_{1:n}^{1:N_x})\in\Theta\times\mathsf{X}^{(n+1)N_x}
\times \{1,\dots,N_x\}^{n}=\mathsf{E}_n$, where
$$
\Psi_{\theta}(x_{0:n}^{1:N_x},a_{1:n}^{1:N_x}) = \Big(\prod_{i=1}^{N_x}\nu_{\theta}(x_0^i)\Big)\prod_{k=1}^n \Big(\prod_{i=1}^{N_x}\frac{g_{\theta}(y_{k-1}|x_{k-1}^{a_k^i})}
{\sum_{j=1}^{N_x}g_{\theta}(y_{k-1}|x_{k-1}^{j})} f_{\theta}(x_k^i|x_{k-1}^{a_k^i})
\Big)
$$
with $g(y_0|x) = 1$ for any $x\in\mathsf{X}$. Throughout, we write the normalizing constant
of $\pi_n$ as $Z_n$.
We make the following defintion for notational convenience in the sequel:
\begin{eqnarray}
\beta_k(\theta,x_{k-1:k}^{1:N_x},a_k^{1:N_x}) & = & \Big(\prod_{i=1}^{N_x}\frac{g_{\theta}(y_{k-1}|x_{k-1}^{a_k^i})}
{\sum_{j=1}^{N_x}g_{\theta}(y_{k-1}|x_{k-1}^{j})} f_{\theta}(x_k^i|x_{k-1}^{a_k^i})
\Big)\Big(\frac{1}{N_x}\sum_{i=1}^{N_x} g_{\theta}(y_k|x_k^i)\Big)\label{eq:beta_def} \\
\check{\beta}_k(\theta,x_{k-1:k}^{1:N_x},a_k^{1:N_x}) & = & \Big(\prod_{i=1}^{N_x}\frac{g_{\theta}(y_{k-1}|x_{k-1}^{a_k^i})}
{\sum_{j=1}^{N_x}g_{\theta}(y_{k-1}|x_{k-1}^{j})} f_{\theta}(x_k^i|x_{k-1}^{a_k^i})
\Big).\label{eq:beta_def1}
\end{eqnarray}

We provide a Feynman-Kac representation of the SMC$^2$ algorithm which facilitates theoretical analysis and will assist our subsequent notations. 
Let $\alpha_k=(\theta,x_{0:k+1}^{1:N_x},a_{1:k+1}^{N_x})$ and set
$$
\eta_0(\alpha_0) = \pi_{\theta}(\theta)\Big(\prod_{i=1}^{N_x}\nu_{\theta}(x_0^i)\Big)\Big(\prod_{i=1}^{N_x}\frac{g_{\theta}(y_{0}|x_{0}^{a_1^i})}
{\sum_{j=1}^{N_x}g_{\theta}(y_{0}|x_{0}^{j})}\Big) f_{\theta}(x_1^i|x_{0}^{a_1^i})
\Big) 
$$
and for any $k\geq 1$
$$
M_k(\alpha_{k-1},d\alpha_k) = \int_{\mathsf{E}_{k+1}} \bar{M}_k(\alpha_{k-1},d\bar{\alpha}_k) \Big(\prod_{i=1}^{N_x}\frac{g_{\theta}(y_{k}|x_{k}^{a_k^i})}
{\sum_{j=1}^{N_x}g_{\theta}(y_{k}|x_{k}^{j})}f_{\theta}(x_{k+1}^i|x_{k}^{a_{k+1}^i})\delta_{\bar{\alpha}_k}(d\xi_k)
\Big)
$$
where $\bar{M}_k$ is $\pi_k-$invariant (i.e.~typically a particle marginal Metropolis-Hastings (PMMH) kernel \cite{Andrieu_2010}). Then we set
$$
G_k(\alpha_k) = \frac{1}{N_x}\sum_{i=1}^{N_x} g_{\theta}(y_{k+1}|x_{k+1}^i).
$$
and
$$
\gamma_n(\varphi) = \int \varphi(\alpha_n)\prod_{p=0}^{n-1} G_p(\alpha_k) \eta_0(\alpha_0)\prod_{k=1}^n M_k(\alpha_{k-1},d\alpha_k) d\alpha_0.
$$
The SMC algorithm which
approximates the $n-$time marginal
$$
\eta_n(\varphi) = \frac{\gamma_n(\varphi)}{\gamma_n(1)}
$$
will have joint law:
\begin{equation}
\prod_{i=1}^N \eta_0(\alpha_0^i)d\alpha_0^i \prod_{k=1}^n\Bigg(
\prod_{i=1}^N  \Phi_k(\eta_{k-1}^N)(d\alpha_k^i) \Bigg)\label{eq:smc_law}
\end{equation}
where $\Phi_k(\mu)(dx) = \mu(G_{k-1}M_k(dx))/\mu(G_{k-1})$ is the usual selection-mutation operator and $\eta_{k-1}^N$ is the empirical measure of the particles. It is easily shown
that by approximating $\eta_n$ one can approximate the posterior on $\theta$ and $x_{0:n}$ given $y_{1:n}$; see \cite{chopin,delmoral2}. We note that, typically one will apply the kernel $\bar{M}_k$
when dynamic resampling is performed; however the form of the algorithm is simple to describe with resampling at each time step.

\subsection{Perfect Algorithm}

One issue with the above algorithm is that whenever the kernel $\bar{M}_k$ is applied, one must sample trajectories of the hidden states which grow with the time parameter. Thus, even if
resampling is dynamically performed, leading to an application of $\bar{M}_k$, the cost of the algorithm will increase with the time parameter. So one can say whilst SM$\textrm{C}^2$ is a very
powerful algorithm, it is not an online algorithm. Here we present an `ideal' or perfect version of the algorithm that has computational cost which does not grow with time, but cannot be implemented in general.
This algorithm will provide the basis for our biased algorithm in the next section.

To begin, we set $B\in\mathbb{Z}^+$ which in principle can grow and $T\in\mathbb{Z}^+$ which is fixed. The parameter $T$ will represent the maximum length of the trajectory of the hidden state,
which one wants to sample. We define the following target probabilities:
\begin{eqnarray}
\pi_{n,T}(\theta,x_{0:n}^{1:N_x},a_{1:n}^{1:N_x}) & \propto &   \pi_{\theta}(\theta) \Big(\prod_{i=1}^{N_x}\nu_{\theta}(x_0^i)\Big)\prod_{k=1}^n \beta_k(\theta,x_{k-1:k}^{1:N_x},a_k^{1:N_x}) \label{eq:smc2target}\\ & & \quad n\in\{0,\dots,T\} \nonumber
\end{eqnarray}
and
$$
\pi_{n,bT}(\theta,x_{(b-1)T+1:n}^{1:N_x},a_{(b-1)T+2:n}^{1:N_x}) \propto 
$$
\begin{equation}
\Big(\prod_{i=1}^{N_x}\zeta(\theta,x_{(b-1)T+1}^i)\Big)
\Big(\frac{1}{N_x}\sum_{i=1}^{N_x} g_{\theta}(y_{(b-1)T+1}|x_{(b-1)T+1}^i)\Big)
\prod_{k=(b-1)T+2}^n \beta_k(\theta,x_{k-1:k}^{1:N_x},a_k^{1:N_x}) \label{eq:smc2target_perf} 
\end{equation}
$$
 n\in\{(b-1)T+1,\dots,bT\},
b\in\{2,\dots,B\}
$$
where $\beta_k$ is as \eqref{eq:beta_def} and 
$$
\zeta(\theta,x_{(b-1)T+1}) \propto \pi_{\theta}(\theta)\int  f_{\theta}(x_{(b-1)T+1}|x_{(b-1)T}) \nu_{\theta}(x_0)\prod_{k=1}^{(b-1)T} f_{\theta}(x_k|x_{k-1}) g_{\theta}(y_k|x_k) dx_{0:(b-1)T}.
$$
If one can approximate the targets above one can approximate the posterior on $\theta$ and $x_{0:n}$ given $y_{1:n}$. As we will see below, our algorithm to achieve this cannot be implemented in practice,
but has the benefit that the computational cost per-time step cannot grow beyond a given bound.

\subsubsection{Algorithm}

For $n\in\{0,\dots,T\}$ one can run the SM$\textrm{C}^2$ algorithm as in Section \ref{sec:smc2}. That is to run the algorithm with law \eqref{eq:smc_law} until time $T-1$. We will add a final time
step that will resample the $N$ particles according to $G_{T-1}$. That is, defining $\check{M}_T$ as the Dirac measure, we sample from
$$
\prod_{i=1}^N \eta_0(\alpha_0^i)d\alpha_0^i \prod_{k=1}^{T}\Bigg(
\prod_{i=1}^N  \Phi_k(\eta_{k-1}^N)(d\alpha_k^i) \Bigg).
$$

For the subsequent blocks, we make the following definitions. Let $b\in\{2,\dots,B\}$, $\check{\alpha}_{(b-1)T+1}=(\theta,x_{(b-1)T+1}^{1:N_x})\in\Theta\times\mathsf{X}^{N_x}=\check{\mathsf{E}}_{(b-1)T+1}$,
$\check{\alpha}_{n}=(\theta,x_{(b-1)T+1:n+1}^{1:N_x},a_{(b-1)T+2:n+1}^{1:N_x})\in\Theta\times\mathsf{X}^{(n-(B-1)T+1)N_x}\times\{1\dots,N_x\}^{(n-(B-1)T)N_x}=\check{\mathsf{E}}_{n}$, $n\in\{(b-1)T+1,\dots,bT-1\}$, and
$\check{\alpha}_{bT}=(\theta,x_{(b-1)T+1:bT}^{1:N_x},a_{(b-1)T+2:bT}^{1:N_x})\in\Theta\times\mathsf{X}^{TN_x}\times\{1\dots,N_x\}^{(T-1)N_x}=\check{\mathsf{E}}_{bT}$. Now let
$$
\check{G}_{n}(\check{\alpha}_{n}) = \frac{1}{N_x}\sum_{i=1}^{N_x} g_{\theta}(y_{n}|x_{n}^i) \quad n\in\{(b-1)T+1,\dots,bT\}.
$$ 
Set
$$
\check{\eta}_{(b-1)T+1}(\check{\alpha}_{(b-1)T+1}) =  \Big(\prod_{i=1}^{N_x} \zeta(\theta,x_{(b-1)T+1}^i)\Big).
$$
Define $\tilde{M}_{n,bT,\pi_{n-1,bT}}$ as a Markov kernel of invariant density $\pi_{n-1,bT}$, $n\in\{(b-1)T+2,\dots,bT-1\}$ such as a PMMH kernel. Then define
$$
\check{M}_{n,bT,\pi_{n-1,bT}}(\check{\alpha}_{n-1},d\check{\alpha}_n) = \tilde{M}_{n,bT,\pi_{n-1,bT}}(\check{\alpha}_{n-1},d\check{\alpha}_{n-1}')\check{\beta}_{n+1}(\theta',(x_{n:n+1}^{1:N_x})',(a_{n+1}^{1:N_x})')d(x_{n:n+1}^{1:N_x})'
$$
where $\check{\alpha}_n=(\check{\alpha}_{n-1}',(x_{n:n+1}^{1:N_x})',(a_{n+1}^{1:N_x})')$ and $\check{\beta}$ is as \eqref{eq:beta_def1}. 
Set $\check{M}_{bT,bT,\pi_{bT-1,bT}}$ as a Dirac mass.
Finally, for $\mu\in\mathscr{P}(\check{E}_{n-1})$, $n\in\{(b-1)T+2,\dots,bT\}$ and probability density $\psi$ on $\check{E}_{n-1}$ define
$$
\check{\Phi}_{n,bT,\psi}(\mu)(d\check{\alpha}_n) := \frac{\mu(\check{G}_{n-1}\check{M}_{n,bT,\psi}(d\check{\alpha}_n))}{\mu(\check{G}_{n-1})}.
$$

Then the perfect algorithm has joint law for $b\in\{2,\dots,B\}$, $n\in\{(b-1)T+1,\dots,bT\}$
$$
\prod_{i=1}^N \check{\eta}_{(b-1)T+1}(\check{\alpha}_{(b-1)T+1}^i) \prod_{k=(b-1)T+2}^n \prod_{i=1}^N \check{\Phi}_{k,,bT,\pi_{k-1,bT}}(\check{\eta}_{k-1,bT}^N)(d\check{\alpha}_k^i)
$$
where $\check{\eta}_{k-1,bT}^N$ is the empirical measure of the particles at time $k-1$.

\subsubsection{Remarks}

Set, for $b\in\{2,\dots,B\}$, $n\in\{(b-1)T+2,\dots,bT\}$
$$
\check{\gamma}_{n,bT}(d\check{\alpha}_n) = \int_{\check{E}_{(b-1)T+1}\times\cdots\check{E}_{n-1}} \prod_{k=(b-1)T+1}^{n-1} \check{G}_k(\check{\alpha}_k)  \check{\eta}_{(b-1)T+1}(d\check{\alpha}_{(b-1)T+1})  \times
$$
$$
\prod_{k=(b-1)T+1}^{n} \check{M}_{k,bT,\pi_{k-1,bT}}(\check{\alpha}_{k-1},d\check{\alpha}_k)
$$
and $\check{\eta}_{n,bT}(d\check{\alpha}_n) = \check{\gamma}_{n,bT}(d\check{\alpha}_n)/\check{\gamma}_{n,bT}(1)$, then we note that for $\varphi:\check{\mathsf{E}}_{n-1}\rightarrow\mathbb{R}$, $\pi_{n,bT}-$integrable
$$
\frac{\check{\eta}_{n,bT}(\check{G}_n\varphi)}{\check{\eta}_{n,bT}(\check{G}_n)}=
$$
$$
\int_{\check{E}_{n-1}} \varphi(\theta,x_{(b-1)T+1:n}^{1:N_x},a_{(b-1)T+2:n}^{1:N_x}) \pi_{n,bT}(\theta,x_{(b-1)T+1:n}^{1:N_x},a_{(b-1)T+2:n}^{1:N_x}) d(\theta,x_{(b-1)T+1:n}^{1:N_x}) 
$$
so that one can estimate expectations w.r.t.~$\pi_{n,bT}$ via 
\begin{equation}
\frac{\check{\eta}_{n,bT}^N(\check{G}_n\varphi)}{\check{\eta}_{n,bT}^N(\check{G}_n)}\label{eq:est_post}.
\end{equation}

\subsection{Approximate Algorithm}

The problem with the previous algorithm is that one can seldom evaluate $\zeta(\theta,x_{(b-1)T+1})$ nor sample from it perfectly. We introduce the following algorithm, which will sample from the following targets.
For the first block, one can run the SM$\textrm{C}^2$ algorithm to target the sequence \eqref{eq:smc2target}. At the subsequent blocks, one is unable to evaluate the target, nor sample from the proposals. We 
propose the following approximate targets to replace \eqref{eq:smc2target_perf}:
$$
\hat{\pi}_{n,bT}(\theta,x_{(b-1)T+1:n}^{1:N_x},a_{(b-1)T+2:n}^{1:N_x})  \propto \Big(\prod_{i=1}^{N_x} \frac{1}{N}\sum_{j=1}^N K_{N}(\theta-\theta^j) f_{\theta}(x_{(b-1)T+1}^i|x_{(b-1)T}^j)\Big)\times
$$
\begin{equation}
\Big(\frac{1}{N_x}\sum_{i=1}^{N_x} g_{\theta}(y_{(b-1)T+1}|x_{(b-1)T+1}^i)\Big) \prod_{k=(b-1)T+2}^n \beta_k(\theta,x_{k-1:k}^{1:N_x},a_k^{1:N_x})\label{eq:approx_target_def}
\end{equation}
with $n\in\{(b-1)T+1,\dots,bT\}, b\in\{2,\dots,B\}$, $\theta^{1:N},x_{(b-1)T}^{1:N}$ samples from the algorithm at the previous block, which we shall describe how to obtain and $K_{N}(\theta-\theta')$ a kernel density whose bandwidth may depend on $N$.
At the end of a block of the algorithm, we just take $\theta^{1:N}$, $x_{(b-1)T}^{1:N}$ as the samples we have obtained, taking the first of the $N_x-$tuples of $x_{T}^{1:N_x}$ (as the samples are exchangeable). Using standard SMC theory, which we will expand upon,
one can prove that if the bandwidth of $K$ appropriately depends on $N$ that at least $\hat{\pi}_{n,2T}(\theta,x_{T+1:n}^{1:N_x},a_{T+2:n}^{1:N_x})$ will converge almost surely (in an appropriate sense) to $\pi_{n,2T}(\theta,x_{T+1:n}^{1:N_x},a_{T+2:n}^{1:N_x})$ as $N$ grows,
hence providing the justification of the approximation introduced.

Set
\begin{eqnarray*}
\widehat{\check{\eta}}_{(b-1)T+1}(\check{\alpha}_{(b-1)T+1}) &=& \Big(\prod_{i=1}^{N_x} \frac{1}{N}\sum_{j=1}^N K_{N}(\theta-\theta^j) f_{\theta}(x_{(b-1)T+1}^i|x_{(b-1)T}^j)\Big)
\end{eqnarray*}
Then the approximate algorithm has joint law for $b\in\{2,\dots,B\}$, $n\in\{(b-1)T+1,\dots,bT\}$
$$
\prod_{i=1}^N \widehat{\check{\eta}}_{(b-1)T+1}(\check{\alpha}_{(b-1)T+1}^i) \prod_{k=(b-1)T+2}^n \prod_{i=1}^N \check{\Phi}_{k,bT,\hat{\pi}_{k-1,bT}}(\check{\eta}_{k-1,bT}^N)(d\check{\alpha}_k^i)
$$
where $\check{\eta}_{k-1,bT}^N$ is the empirical measure of the particles at time $k-1$. An estimate of the form \eqref{eq:est_post} can be used to estimate the targets. Note that the cost of the algorithm is not $\mathcal{O}(N^2)$
as one does not need to evaluate $\widehat{\check{\eta}}_{(b-1)T+1}(\check{\alpha}_{(b-1)T+1})$, even in the PMMH steps.

\subsection{Related Simulation Methods and Alternatives}

Similar, but different, ideas have appeared in several articles including \cite{andrieu,cent,dan}. These ideas are considered in the context of hidden
Markov models and partially observed point processes respectively. The key differences of our work to \cite{cent, dan} (\cite{andrieu} is for maximum likelihood estimation (MLE))
are as follows. In the context of \cite{cent} we do not use a type of `sequential MCMC', in that our approach can be used explicitly for online Bayesian parameter estimation. The approach of
\cite{dan} is less general, where the particles are not updated with PMCMC, and one has a block length of 1.

Alternatives to the approach outlined above are; (i) a form of sequential PMCMC in the spirit of \cite{cent} or (ii) to use an over-lapping or sliding window. For (i), one expects that the cost is higher than the above algorithm,
and online (filtered) estimates are not available. For (ii) the cost may be higher, but, in simulation studies, we did not find any obvious improvement in practice. We also remark that our approximate algorithm uses kernel density
estimation, but, in principle, any approximation scheme could be used; this is demonstrated in Section \ref{sec:simos} where the method in \cite{chan} is used in place of a kernel density estimate. In practice, the algorithm is implemented
with dynamic resampling according to the effective sample size.

\section{Theoretical Results}\label{sec:theory}

Throughout, it is supposed that for any $n\geq 1$, $\sup_{\theta,x}g_{\theta}(y_n|x)<+\infty$.

\subsection{Consistency}

We will now show that, under a specific choice of the bandwidth $h$ and under some mathematical assumptions, the algorithm just presented is consistent.
We set $K_h(\theta) = h^{-d}K(h^{-1}\theta)$; it is supposed that for any fixed $\theta$, $\lim_{h\rightarrow 0} K_h(\theta) = 0$. For simplicity we denote $\eta_{n}(x)$ (resp.~$\mathsf{E}_n$), $n\in\{0,\dots,T\}$ as $\check{\eta}_{n,T}(x)$ (resp.~$\check{\mathsf{E}}_n$).

\begin{hypA}
We have 
\begin{itemize}
\item{$\int_{\Theta}K(\theta)d\theta=1$, $K(\theta)\geq 0~\forall \theta\in\Theta$}
\item{$\int\|\theta\|^2K(\theta)d\theta<+\infty$}
\item{$K\in\mathcal{C}_b(\Theta)$.}
\end{itemize}
\end{hypA}
\begin{hypA}
For each $n,b$, $b\in\{1,\dots,B\}$, $n\in\{(b-1)T+1,\dots,bT\}$, there exists a $L_{n,k}>0$ such that for every $\theta,\theta'\in\Theta$
$$
|\check{\eta}_{n,bT}(\theta)-\check{\eta}_{n,bT}(\theta')| \leq  L_{n,k}\|\theta-\theta'\|.
$$
\end{hypA}
\begin{hypA}\label{hyp:3}
We have $h=N^{-\frac{1}{2(d+1)}}$, in addition there exist a $C>0$ such that for every $N$ $\sup_{\theta\in\Theta}|K_N(\theta)|\leq C$.
\end{hypA}
\begin{hypA}\label{hyp:4}
For any $b\in\{2,\dots,B\}$, $n\in\{(b-1)T+1,\dots,bT\}$, $\varphi\in\mathcal{B}(\check{\mathsf{E}}_{n})$ the function $\eta\mapsto \check{M}_{n,bT,\eta}^n(\varphi)(x)$ is continuous at 
$\eta=\pi_{n-1,bT}$ uniformly in $x\in \check{\mathsf{E}}_{n-1}$; also $\sup_{\eta,x}\check{M}_{n,bT,\eta}^n(\varphi)(x)<+\infty$.
\end{hypA}

The first three assumptions are from \cite{crisan} to allow convergence of the SMC estimate of the kernel density and the final assumption is from \cite{beskos}. We set
$$
\check{\eta}_{(b-1)T,(b-1)T}^{N}(K_N(\theta)f_{\theta}(x_{(b-1)T+1}|\cdot)) = \frac{1}{N}\sum_{j=1}^N K_{N}(\theta-\theta^j) f_{\theta}(x_{(b-1)T+1}|x_{(b-1)T}^j).
$$
Let $\rightarrow_{\mathbb{P}}$ denote convergence in probability as $N\rightarrow\infty$. We now have the following result whose proof is in Appendix \ref{app:a}.

\begin{theorem}\label{theo:consis}
Assume (A1-4). Then for each $n,b$, $b\in\{1,\dots,B\}$, $n\in\{(b-1)T+1,\dots,bT\}$ and
$\varphi\in\mathcal{B}_b(\check{\mathsf{E}}_n)$ we have
$$
\check{\eta}_{n,bT}^{N}(\varphi) \rightarrow_{\mathbb{P}} \check{\eta}_{n,bT}(\varphi)
$$
and in particular for every $\theta\in\Theta, x_{(b-1)T+1}\in\mathsf{X}$
$$
\check{\eta}_{(b-1)T,(b-1)T}^{N}(K_N(\theta)f_{\theta}(x_{(b-1)T+1}|\cdot)) \rightarrow_{\mathbb{P}} \zeta(\theta,x_{(b-1)T+1}).
$$
\end{theorem}

\begin{rem}
The result here is essentially qualitative. In order to obtain consistency, one must adopt an exponential effort in $d$ and so one would only run (with the choice of $h$ as in (A\ref{hyp:3})) over the exact algorithm
if the time horizon is long and $d$ is small. In practice one will decouple $h$ and $N$, leading to estimates which are biased, even if $N\rightarrow\infty$. We will now study the bias.
\end{rem}

\subsection{Bias}

We now investigate the asymptotic (in $N$) bias of the approach. We make the following hypothesis:

\begin{hypA}\label{hyp:bias}
There exists a $\delta\in[1,\infty)$  such that for any $n\geq T+1$:
$$
\sup_{(x,y)\in\check{\mathsf{E}}_{n}} \frac{\check{G}_n(x)}{\check{G}_n(y)} \leq \delta.
$$
There exists an $\epsilon\in(0,1)$ and for any $b\in\{2,\dots,B\}$, $n\in\{(b-1)T+2,\dots,bT-1\}$
there exists a $\nu\in\mathscr{P}(\check{\mathsf{E}}_n)$ such that for any probability density $\psi$ on $\check{E}_{n-1}$
$(x,y)\in\check{E}_{n-1}^2$
$$
\check{M}_{n,bT,\psi}(x,\cdot) \geq \epsilon\Big(\check{M}_{n,bT,\psi}(x,\cdot)\wedge \nu(\cdot)\Big).
$$
\end{hypA}
We write the $N-$limiting version of \eqref{eq:approx_target_def} as $\check{\pi}_{n,bT,\infty}$ (which one can easily prove exists, for any $b\in\{2,\dots,B\}$).
For a bounded $\varphi:\mathsf{E}_{bT}\rightarrow\mathbb{R}$ we define the bias as
$$
\mathsf{B}(bT,\varphi) = \Big|\Big\{\mathbb{E}_{\pi_{bT,bT}} - \mathbb{E}_{\check{\pi}_{n,bT,\infty}}\Big\}[\varphi(X_{bT}^{1:N_x},a_{bT}^{1:N_x},\theta)]\Big|.
$$
We have the following result whose proof is in Appendix \ref{app:bias}.

\begin{theorem}\label{theo:bias}
Assume (A\ref{hyp:bias}). Then there exist a $C<+\infty$, $\rho_1,\rho_2\in (0,1)$ such that for any $b\in\{2,\dots,B\}$:
$$
\mathsf{B}(bT,\varphi) \leq C\|\varphi\|_{\infty}(1-\rho_1^{T-1} + \rho_2^{T-1}).
$$
\end{theorem}

\begin{rem}
The result indicates that some of the bias can fall at a geometric rate as $T$ grows. Note, that the bias cannot disappear once, one block is wrong (which is by assumption in the statement of the theorem).
 The result also provides the reassuring point that bias' do not accumulate as the number of blocks grow, albeit under strong assumptions.
\end{rem}

\section{Simulations}\label{sec:simos}

\subsection{Gaussian linear model}

\subsubsection{Model and algorithm setting}

We consider the following Gaussian model,
\begin{align*}
  \nu_{\theta}(x_0) &= \phi(x_0; 0, \tau_0^{-1}) \\
  f_{\theta}(x_k|x_{k-1}) &= \phi(x_k; x_{k-1}, \tau^{-1}) \\
  g_{\theta}(y_k|x_k) &= \phi(y_k; x_k, \lambda^{-1})
\end{align*}
where $\phi(x; \mu, \sigma^2)$ is the density function of a Normal
distribution with mean $\mu$ and variance $\sigma^2$.

The data is simulated from $\tau_0 = \tau = \lambda = 1$, with 10,000 time
steps. The algorithms are implemented with $\tau_0$ fixed to $1$ and the other
two being estimated. That is, the parameter is $\theta = (\tau, \lambda)$.

We consider three algorithms. First, due to the simple structure of the model,
we can obtain exact evaluation of $p(y_k|y_{1:k-1})$ through a Kalman filter. We will replace the particle filter in the SM$\textrm{C}^2$
algorithm with the Kalman filter; this results in a regular SMC algorithm. This is
used to provide accurate unbiased estimators when comparing algorithms. The
second is the SM$\textrm{C}^2$  algorithm and the third is the proposed new algorithm, which is termed
SM$\textrm{C}^2$FW.  The PMCMC step in each algorithm is constructed as a two-block
Metropolis random walk, with a Normal kernel on the logarithm scale.

With the Kalman filter, the computational cost of $p(y_k|y_{1:k-1})$ is
negligible and thus we use a large number of $\theta$-particles, $N=
10000$. This allows us to obtain unbiased posterior estimators with very small
variance. As a result, this provides a good baseline for the comparison of the other two
algorithms.

For the SM$\textrm{C}^2$ algorithm, we use $N = 1000$ particles for the SMC
algorithm and $N_x = 1000$ particles. For simplicity, the value of $N_x$ is fixed through the time line.

For the SM$\textrm{C}^2$FW algorithm, we set $N = 10,000$ and $N_x = 1000$. As
we will see later, despite the increased number of the $\theta$-particles, the
computational cost is still significantly lower than SM$\textrm{C}^2$, in addition to
be upper bounded per time step. The kernel density estimate (KDE) is bivariate Normal on
the logarithm scale, with covariance matrix taken a diagonal form, $\Sigma =
hI_2$ where $I_2$ is the identity matrix of rank two. We consider $h =
0.01$, $0.05$ and $0.25$. Three widths of each fixed window $T = 125$, $500$ and
$1000$ are also considered.

\subsubsection{Results}

In Figure~\ref{fig:gl_mean_x} to~\ref{fig:gl_mean_lambda} we show the average
of estimates over 30 simulations, given different bandwidth $h$ and window
width $T$. The first 125 time steps are cutoff from the graphs, during which
time the SM$\textrm{C}^2$FW algorithm is exactly the same as the SM$\textrm{C}^2$ algorithm.

\begin{figure}[t]
  \includegraphics[width=\linewidth]{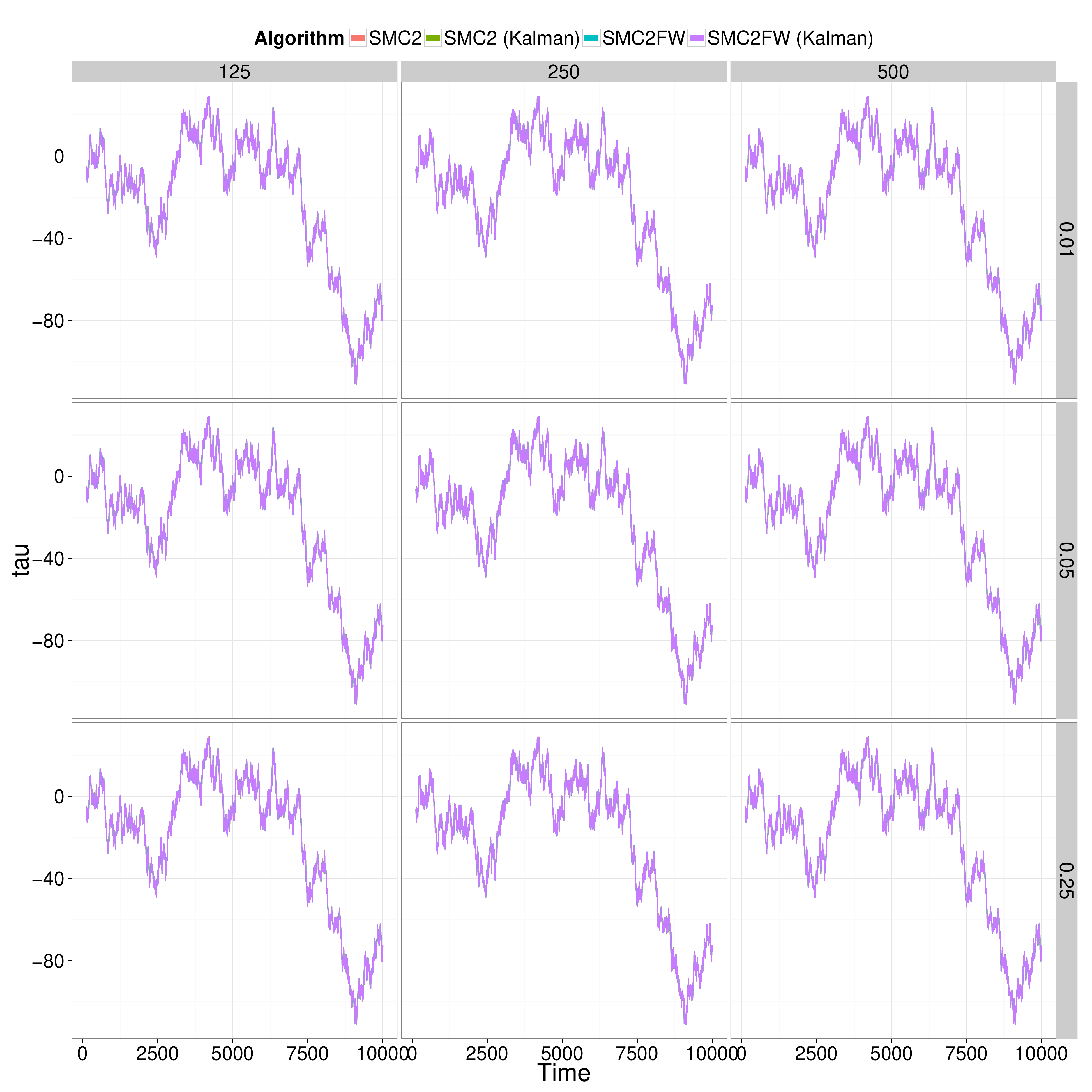}
  \caption{Average of estimates of state $X_k$ for the Gaussian linear model.}
  \label{fig:gl_mean_x}
\end{figure}

\begin{figure}[t]
  \includegraphics[width=\linewidth]{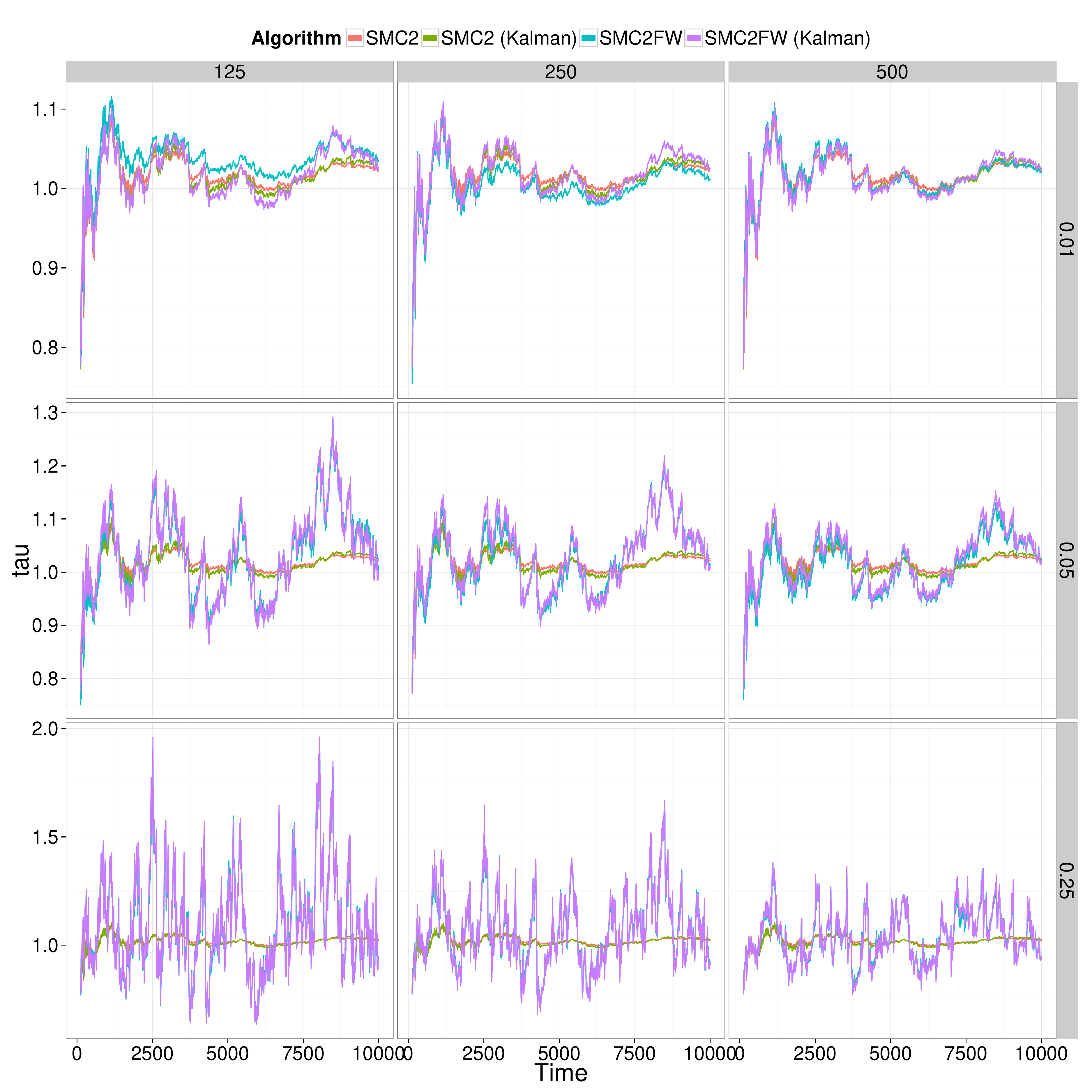}
  \caption{Average of estimates of parameter $\tau$ (transition precision) for
   the Gaussian linear model.}
  \label{fig:gl_mean_tau}
\end{figure}

\begin{figure}[t]
  \includegraphics[width=\linewidth]{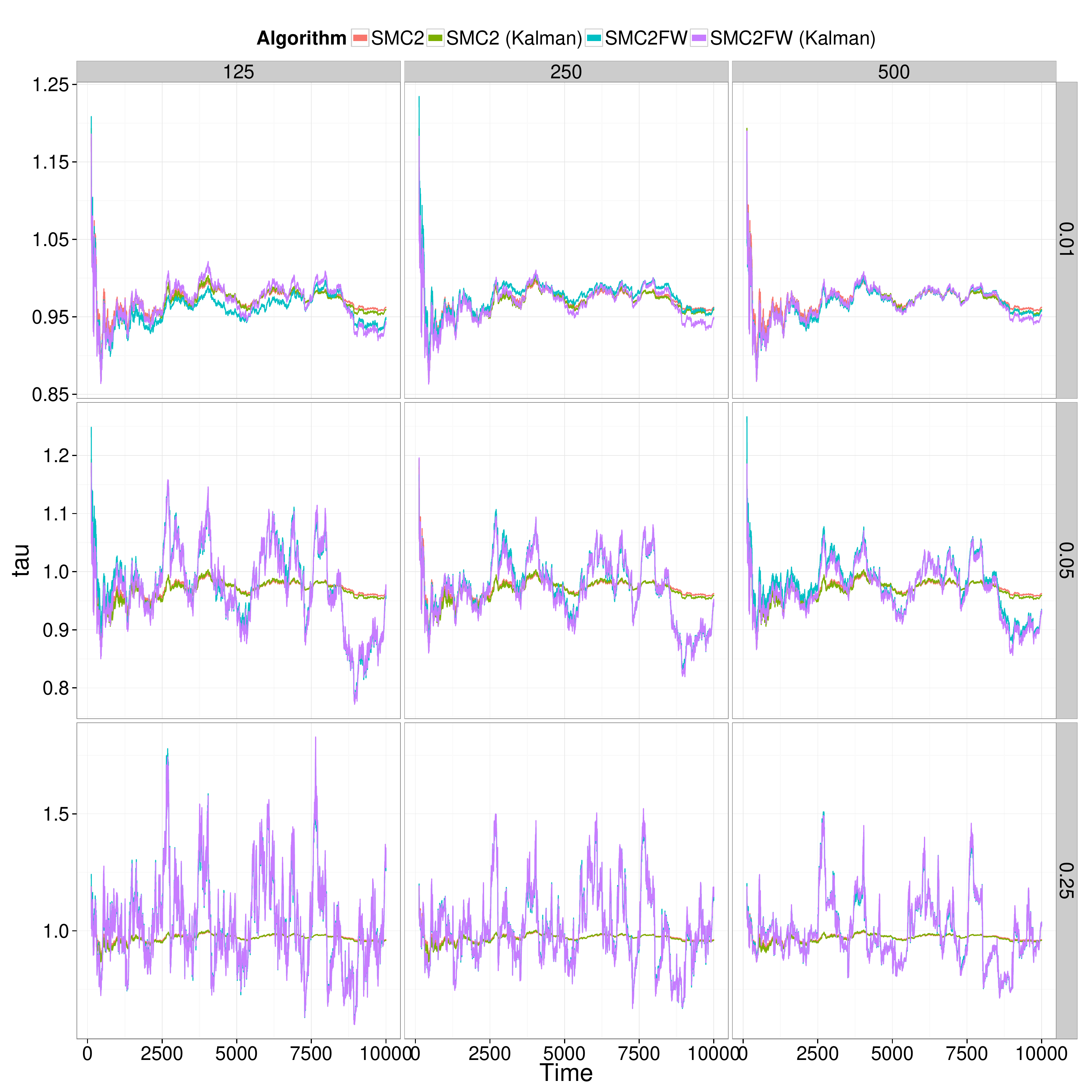}
  \caption{Average of estimates of parameter $\lambda$ (observation
    precision) for the Gaussian linear model.}
  \label{fig:gl_mean_lambda}
\end{figure}

Algorithms labeled with ``Kalman'' means a Kalman filter is used in place of a
particle filter and thus exact values of $p(y_k|y_{1:k-1})$ are calculated
instead of approximations. In this case, the results show the behavior of the
SM$\textrm{C}^2$FW algorithm when $N_x\to\infty$. However, in this particular example,
errors introduced by the particle filter approximation is minimal as shown in
the graphs.

It is as expected that the longer the window width, the better the SM$\textrm{C}^2$FW
algorithms perform. The choice of the bandwidth $h$ has a more dramatic effect
on the performance. With $h = 0.01$ the algorithms give almost exact results
even for $T = 125$. On the other hand, with $h = 0.25$, the errors are
numerous. For the state $X_k$, neither the bandwidth nor the window width
affect the results in any observable way. For the parameters, $h = 0.01$ and
$T = 500$ provides the best results.

In Figure~\ref{fig:gl_mse} the MSE of the two algorithms are shown, using the
results obtained with a SMC algorithm using Kalman filters (SM$\textrm{C}^2$ (Kalman)
in previous figures) as an unbiased, accurate estimate of the true posterior
means. It can be seen that for the state $X_k$, the results are mixed. For the
$\tau$ parameter the SM$\textrm{C}^2$ algorithm has a smaller MSE while the opposite
is true for the $\lambda$ parameter. It shall be noted that, the SM$\textrm{C}^2$FW
algorithm use $N = 10,000$ while the SM$\textrm{C}^2$ algorithm only use
$N = 1000$. However, despite the large difference of the numbers of
$\theta$-particles, the SM$\textrm{C}^2$FW algorithm is still significantly more
computationally cost efficient. It took about thirty minutes for a single
run of the SM$\textrm{C}^2$FW algorithm under this setting while it took about five
hours for the SM$\textrm{C}^2$ algorithm. More importantly, the cost of SM$\textrm{C}^2$FW is
bounded per time step, and thus it is possible to obtain better results than
SM$\textrm{C}^2$ for all parameters while having a bounded, smaller cost in the long
run.

\begin{figure}[t]
  \includegraphics[width=\linewidth]{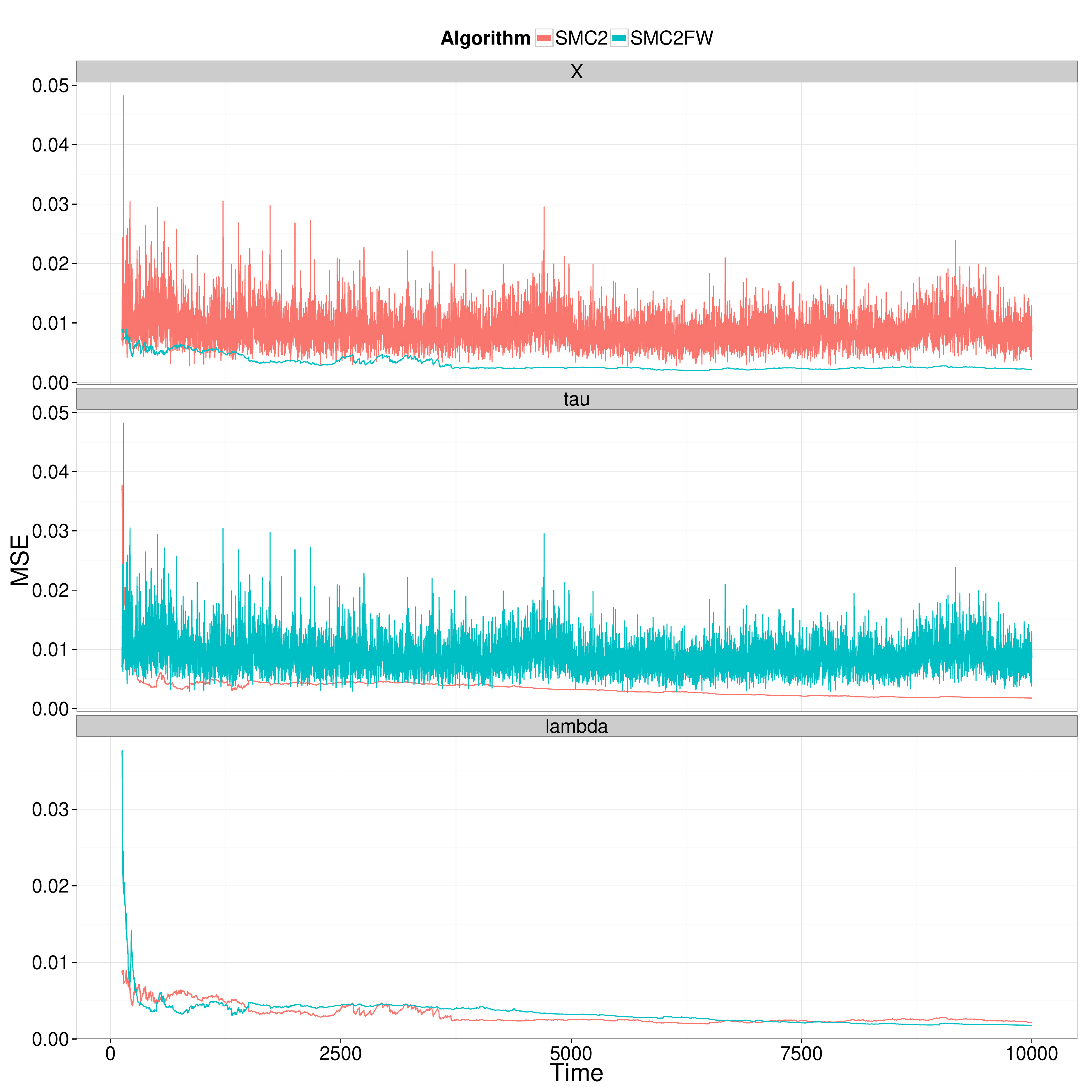}
  \caption{MSE of estimates of parameters and states for the Gaussian linear
    model. The SM$\textrm{C}^2$FW algorithm uses $h = 0.01$ and $l = 500$.}
  \label{fig:gl_mse}
\end{figure}

\subsection{\protect\levy-driven stochastic model}

\subsubsection{Model and algorithm setting}

We consider the \levy-driven stochastic volatility model, applied to 1,000
recent S\&P 500 data from September 17, 2010 to September 8, 2014. The data
are obtained as logarithm return of the daily adjusted close price, and then
normalized to unity variance. The data is plotted in
Figure~\ref{fig:levy_data}. There are considerably much larger volatility at
the beginning of the series. It becomes much more stable in the middle and
slightly larger at the end.

\begin{figure}[t]
  \includegraphics[width=\linewidth]{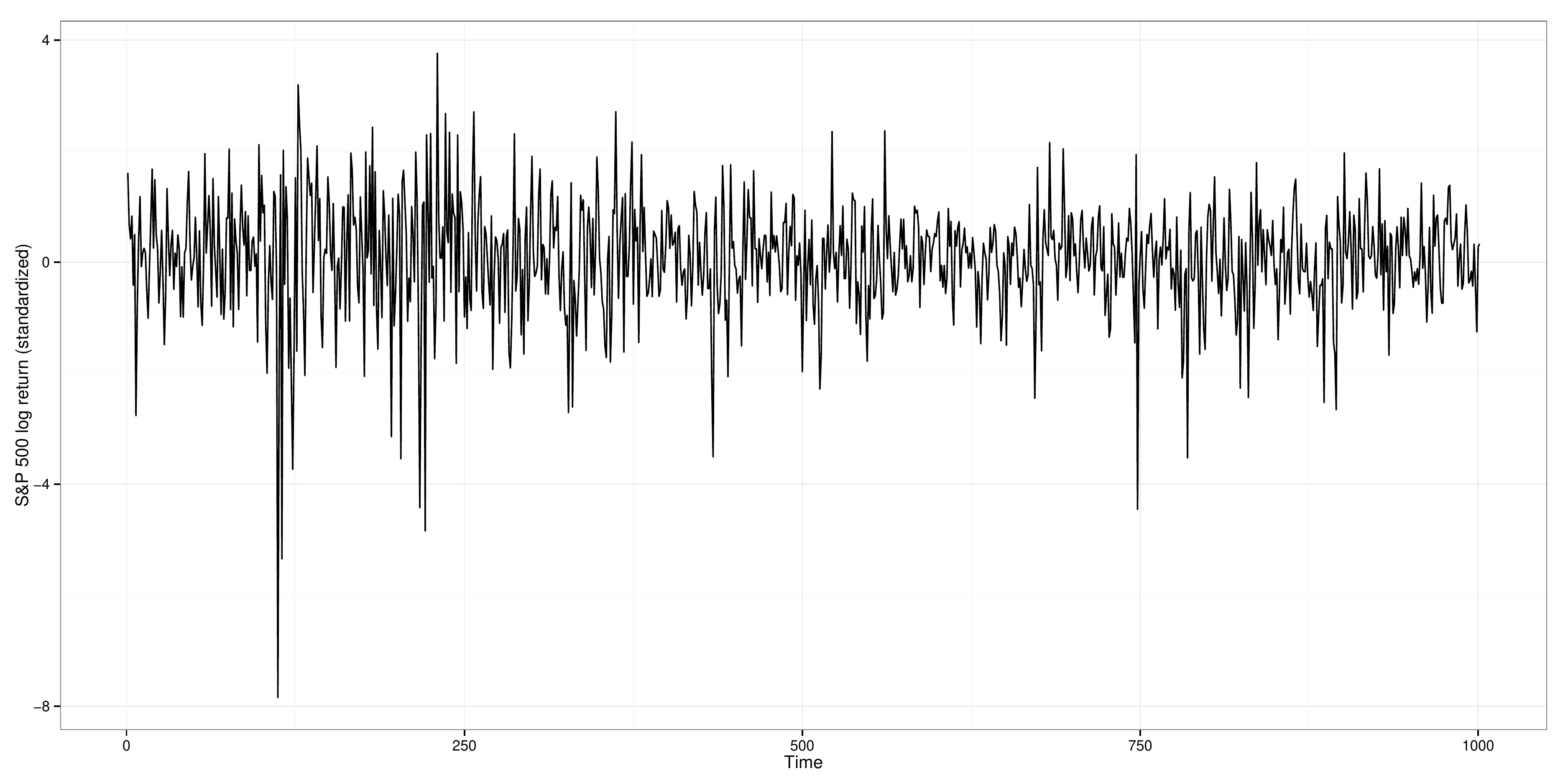}
  \caption{S\&P 500 daily log-return data.}
  \label{fig:levy_data}
\end{figure}

The model we employed to analyze this data is the same as in
\cite{Andrieu_2010} and we use the same notations and
formulation as in that paper. The model has four parameters $\theta = (\kappa,
\delta, \gamma, \lambda)$ and a two-dimensional state $X_n =
(\sigma^2(n\Delta), z(n\Delta))$ where $\Delta$ is the time interval and in
this example it set to constant $1$. 

We consider three algorithms. First, the SM$\textrm{C}^2$ algorithm with $N =
1000$ and $N_x = 1000$. Second, the SM$\textrm{C}^2$FW algorithm with the same number
of particles. The same Normal KDE approximation is used as in the last
example. Various bandwidths of the kernel were considered, and $h = 0.01$ is
chosen. The last, the idea coupled with the parallel particle filter
algorithm (PPF) \cite{chan}, instead of the KDE
approximation, is considered. The latter two algorithms use a window width $T =
200$. We found that further increasing the window width does not improve
results in a significant manner.

The algorithms again use Normal kernels on logarithm scales for the PMMH
proposals. Using results from the particle filter paper, we believe that
$\kappa$ and $\delta$ are strongly correlated and they are updated in one
block of PMCMC move while $\gamma$ and $\lambda$ are updated individually in
their own block. The proposal scales are calibrated on-line using the moments
estimates from the sampler.

\subsubsection{Results}

In Figure~\ref{fig:levy_avg} we show the average of estimates over 20
simulations. All three algorithms give similar results. Compared to results
from the SM$\textrm{C}^2$ algorithm, the SM$\textrm{C}^2$FW algorithm give slightly better
result than that of the SM$\textrm{C}^2$FW-PPF one.

\begin{figure}[t]
  \includegraphics[width=\linewidth]{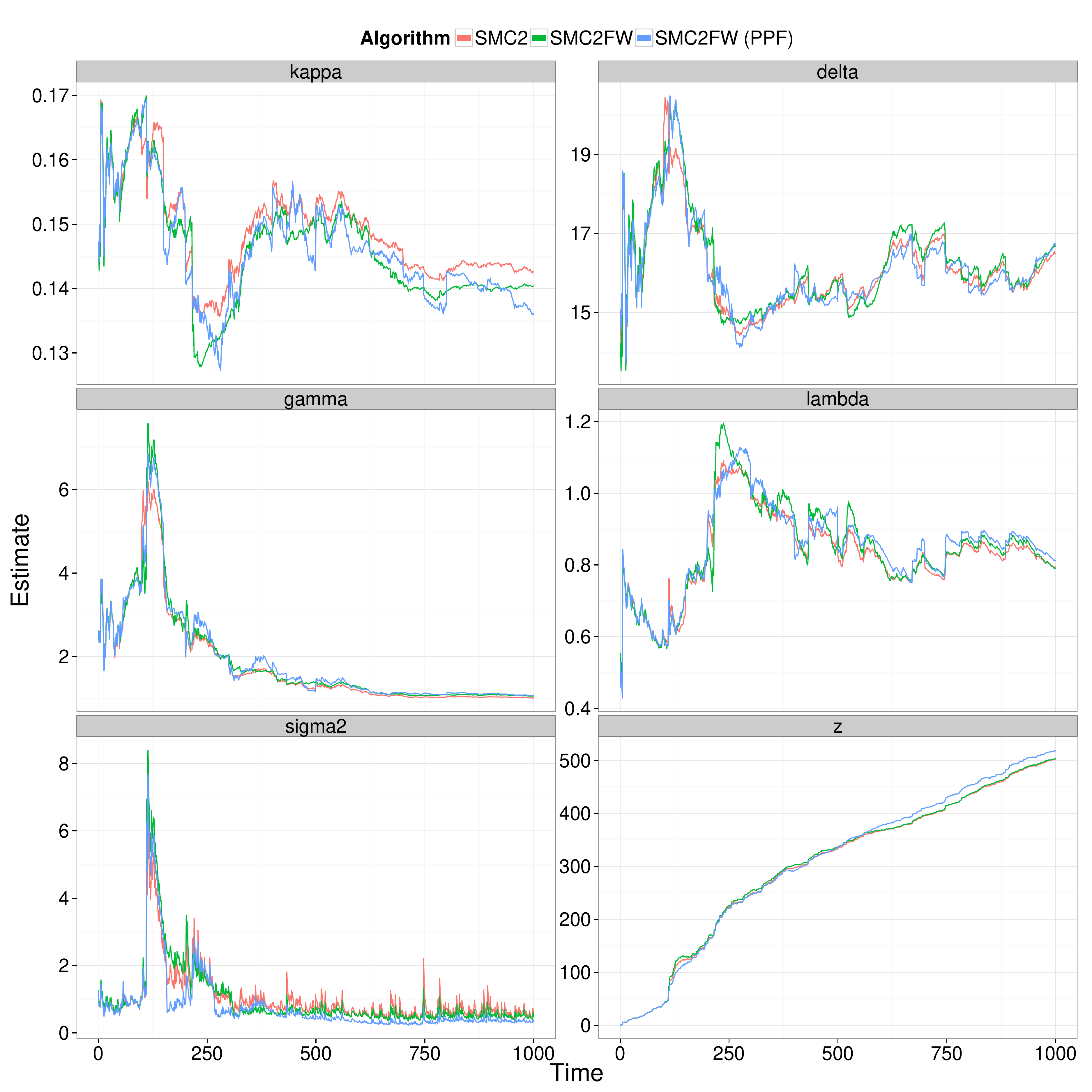}
  \caption{Averages of estimates for the \levy-driven stochastic model using
    S\&P 500 data.}
  \label{fig:levy_avg}
\end{figure}

One feature that was mentioned in Section \ref{sec:theory}, but
not shown clearly in the early simple Gaussian linear example, is that the
bias of the online algorithms does not accumulate over time. For instance,
consider the $\sigma^2$ state, though non-trivial errors can be observed for
both on-line algorithms in the early time steps, they were not carried on to
later times.

In addition to the estimates, prediction of the variance (square of
volatility), $\sigma^2(n\Delta)$ is also calculated as $\Exp[X_n|y_{1:n-1}] =
\Exp[\Exp[X_n|X_{n-1}]|y_{1:n-1}]$. Since the transition density is not of a
closed form, its expectation is estimated with 100 samples of $X_n$ generated
for each value of $X_{n-1}$ in the particle system. The outer expectation
is approximated with the particle system at time $n - 1$. The results for the three algorithms are plotted
in Figure~\ref{fig:levy_vol} against the squared log-returns.


\begin{figure}[t]
  \includegraphics[width=\linewidth]{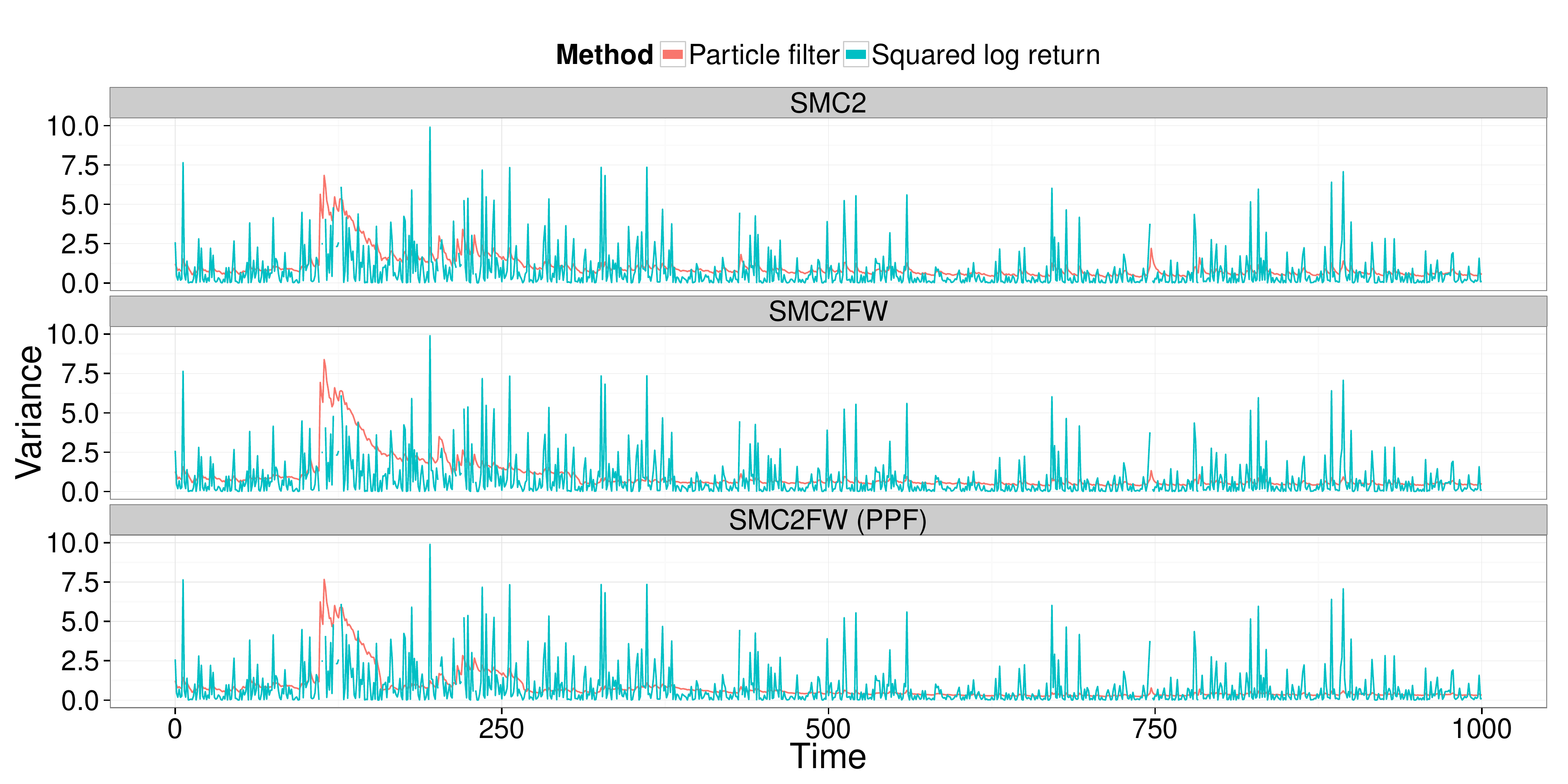}
  \caption{Prediction of variance for the \levy-driven stochastic model using
    S\&P 500 data}
  \label{fig:levy_vol}
\end{figure}


\subsubsection{Summary}

Through two examples, it is shown that with some careful choice of the KDE
bandwidth, or using the parallel particle filter algorithm, which does not
require this layer of tuning and an appropriate window width, it is possible
to obtain results competitive to the SM$\textrm{C}^2$ algorithm with only a
fraction of computational cost. In the \levy-driven stochastic model example,
the theoretical feature that the bias does not accumulate over time is demonstrated.

\section{Summary}\label{sec:summary}

In this article we have presented a method for Bayesian online static parameter estimation for state-space models. Our method is such that
the computational cost does not grow with the time parameter and moreover, the algorithm can be shown to be consistent, although, the cost is
then exponential in the dimension of the static parameter. We have additionally shown that the asymptotic bias, under strong assumptions, does not grow with time.

There are several avenues for future work. First is the use of alternative approximation schemes in our algorithm; we have relied on kernel density estimation, but there are other schemes
which could be used. Second, in our work, we have investigated the asymptotic bias. However, as is clear in the proofs, then the consecutive blocks are independent and one thus expects that
the study of the finite sample bias is significantly more challenging; an investigation of this is warranted.

\subsubsection*{Acknowledgements}
This research was supported by a Singapore Ministry of Education Academic Research Fund Tier 1 grant (R-155-000-156-112).
We thank Alex Beskos \& Alex Thiery for many useful conversations on this work.

%
%

\appendix

\section{Technical Results: Consistency}\label{app:a}

\begin{proof}[Proof of Theorem \ref{theo:consis}]
For $n\in\{0,\dots,T\}$ the result follows by standard theory; see \cite{crisan,delmoral1,delmoral2}. We first consider
$\check{\eta}_{T,T}^{N}(K_N(\theta)f_{\theta}(x_{T+1}|\cdot))$. Denote by $\mathscr{F}_{T}$ the $\sigma-$algebra generated
by the particle system up-to time $T$ and expectations w.r.t.~the law of the simulated algorithm as $\mathbb{E}$. Then we have 
$$
\check{\eta}_{T,T}^{N}(K_N(\theta)f_{\theta}(x_{T+1}|\cdot)) - \zeta(\theta,x_{T+1}) = 
$$
$$
\check{\eta}_{T,T}^{N}(K_N(\theta)f_{\theta}(x_{T+1}|\cdot)) - \mathbb{E}[\check{\eta}_{T,T}^{N}(K_N(\theta)f_{\theta}(x_{T+1}|\cdot))|\mathscr{F}_{T}]
+ \mathbb{E}[\check{\eta}_{T,T}^{N}(K_N(\theta)f_{\theta}(x_{T+1}|\cdot))|\mathscr{F}_{T}] -
\zeta(\theta,x_{T+1}).
$$
For the first term on the R.H.S.~one has by the (conditional) Marcinkiewicz-Zygmund inequality (conditional on $\mathscr{F}_{T}$ the samples are generated independently):
$$
\mathbb{E}[(\check{\eta}_{T,T}^{N}(K_N(\theta)f_{\theta}(x_{T+1}|\cdot)) - \mathbb{E}[\check{\eta}_{T,T}^{N}(K_N(\theta)f_{\theta}(x_{T+1}|\cdot))|\mathscr{F}_{T}])^2] \leq \frac{C'}{\sqrt{N}}
$$
where $C'$ depends on $C$ in (A\ref{hyp:3}); thus $\check{\eta}_{T,T}^{N}(K_N(\theta)f_{\theta}(x_{T+1}|\cdot)) - \mathbb{E}[\check{\eta}_{T,T}^{N}(K_N(\theta)f_{\theta}(x_{T+1}|\cdot))|\mathscr{F}_T]$ converges to zero in probability.

Now, for
$$
\mathbb{E}[\check{\eta}_{T,T}^{N}(K_N(\theta)f_{\theta}(x_{T+1}|\cdot))|\mathscr{F}_{T}] -
\zeta(\theta,x_{T+1}) 
$$
$\mathbb{E}[\check{\eta}_{T,T}^{N}(K_N(\theta)f_{\theta}(x_{T+1}|\cdot))|\mathscr{F}_{T}]=\Phi_T(\eta_{T-1}^N)(K_N(\theta)f_{\theta}(x_{T+1}|\cdot))$, and
the denominator converges in probability and by  the arguments in \cite{crisan}[Theorem 4.1] the numerator will converge in probability to the appropriate quantity; that is
$\Phi_T(\eta_{T-1}^N)(K_N(\theta)f_{\theta}(x_{T+1}|\cdot))- \zeta(\theta,x_{T+1}) \mathbb{P} \rightarrow 0$.

For $\check{\eta}_{T+1,2T}^{N}(\varphi)$, this converges in probability to $\check{\eta}_{T+1,2T}(\varphi)$ by centering by the conditional expectation (given $\mathscr{F}_{T+1}$)
and applying the (conditional) Marcinkiewicz-Zygmund inequality for $\check{\eta}_{T+1,2T}^{N}(\varphi)-\mathbb{E}[\check{\eta}_{T+1,2T}^{N}(\varphi)|\mathscr{F}_{T+1}]$.
The term $\mathbb{E}[\check{\eta}_{T+1,2T}^{N}(\varphi)|\mathscr{F}_{T+1}]$ will converge to $\check{\eta}_{T+1,2T}(\varphi)$ using the conditional i.i.d.~property
and the fact that $\check{\eta}_{T,T}^{N}(K_N(\theta)f_{\theta}(x_{T+1}|\cdot))$ converges in probability to $\zeta(\theta,x_{T+1})$.

For $\check{\eta}_{T+2,2T}^{N}(\varphi)$, we consider:
$$
\check{\eta}_{T+2,2T}^{N}(\varphi) - \Phi_{T+2,2T,\hat{\pi}_{T+1,2T}}(\check{\eta}_{T+1,2T}^{N})(\varphi)  +
\Phi_{T+2,2T,\hat{\pi}_{T+1,2T}}(\check{\eta}_{T+1,2T}^{N})(\varphi) -\check{\eta}_{T+2,2T}(\varphi).
$$
The first term is dealt with via the (conditional) Marcinkiewicz-Zygmund inequality as above. For 
$\Phi_{T+2,2T,\hat{\pi}_{T+1,2T}}(\check{\eta}_{T+1,2T}^{N})(\varphi) -\check{\eta}_{T+2,2T}(\varphi)$
as $\check{\eta}_{T+1,2T}^{N}(\check{G}_{T+1})$ converges in probability to 
$\check{\eta}_{T+1,2T}(\check{G}_{T+1})$, we need only consider
$$
\check{\eta}_{T+1,2T}^{N}(\check{G}_{T+1}\check{M}_{T+2,2T,\hat{\pi}_{T+1,2T}}(\varphi)) -
\check{\eta}_{T+1,2T}(\check{G}_{T+1}\check{M}_{T+2,2T,\pi_{T+1,2T}}(\varphi)) = 
$$
$$
\check{\eta}_{T+1,2T}^{N}(\check{G}_{T+1}\check{M}_{T+2,2T,\hat{\pi}_{T+1,2T}}(\varphi)) - 
\check{\eta}_{T+1,2T}^{N}(\check{G}_{T+1}\check{M}_{T+2,2T,\pi_{T+1,2T}}(\varphi)) +
$$
$$
\check{\eta}_{T+1,2T}^{N}(\check{G}_{T+1}\check{M}_{T+2,2T,\pi_{T+1,2T}}(\varphi)) -
\check{\eta}_{T+1,2T}(\check{G}_{T+1}\check{M}_{T+2,2T,\pi_{T+1,2T}}(\varphi))
$$
The last term on the R.H.S.~converges to zero by the above calculations. So we focus on the first term on the R.H.S.~we have
$$
\mathbb{E}[|\check{\eta}_{T+1,2T}^{N}(\check{G}_{T+1}\check{M}_{T+2,2T,\hat{\pi}_{T+1,2T}}(\varphi)) - 
\check{\eta}_{T+1,2T}^{N}(\check{G}_{T+1}\check{M}_{T+2,2T,\pi_{T+1,2T}}(\varphi))|] \leq
$$
$$
C\mathbb{E}[|[\check{M}_{T+2,2T,\hat{\pi}_{T+1,2T}}(\varphi)(\alpha_{T+1}^1) - \check{M}_{T+2,2T,\pi_{T+1,2T}}(\varphi)(\alpha_{T+1}^1)]|]
$$
Let $\epsilon>0$ be given. By (A\ref{hyp:4}) there exists $\delta>0$ independent of $\alpha_{T+1}^1$ such that for any probability density $\eta$ with $|\eta-\pi_{T+1,2T}|<\delta$ we have that 
$|\check{M}_{T+2,2T,\hat{\pi}_{T+1,2T}}(\varphi)(\alpha_{T+1}^1) - \check{M}_{T+2,2T,\pi_{T+1,2T}}(\varphi)(\alpha_{T+1}^1)|<\epsilon/2$. Consider the event:
$$
A(N,\delta) = \{\,|\hat{\pi}_{T+1,2T}-\pi_{T+1,2T}|<\delta\, \}\ .
$$
Then 
$$
\mathbb{E}[|[\check{M}_{T+2,2T,\hat{\pi}_{T+1,2T}}(\varphi)(\alpha_{T+1}^1) - \check{M}_{T+2,2T,\pi_{T+1,2T}}(\varphi)(\alpha_{T+1}^1)]|] =
$$
$$
\mathbb{E}[|[\check{M}_{T+2,2T,\hat{\pi}_{T+1,2T}}(\varphi)(\alpha_{T+1}^1) - \check{M}_{T+2,2T,\pi_{T+1,2T}}(\varphi)(\alpha_{T+1}^1)]|\mathbb{I}_{A(N,\delta)}] +
$$
$$
\mathbb{E}[|[\check{M}_{T+2,2T,\hat{\pi}_{T+1,2T}}(\varphi)(\alpha_{T+1}^1) - \check{M}_{T+2,2T,\pi_{T+1,2T}}(\varphi)(\alpha_{T+1}^1)]|\mathbb{I}_{A(N,\delta)^c}] \leq
$$
$$
\epsilon/2 + C\mathbb{P}(A(N,\delta)^c).
$$
By the convergence in probability of $\check{\eta}_{T,T}^{N}(K_N(\theta)f_{\theta}(x_{T+1}|\cdot))$ here is an $N_0\geq 1$ such that for each $N\geq N_0$ we have $2 C\cdot \mathbb{P}\,[\,A(N,\delta)^c\,]\leq \frac{\epsilon}{2}$. Hence, for any $N\geq N_0$:
$$
\mathbb{E}[|\check{G}_{T+1}(\alpha_{T+1}^1)[\check{M}_{T+2,2T,\hat{\pi}_{T+1,2T}}(\varphi)(\alpha_{T+1}^1) - \check{M}_{T+2,2T,\pi_{T+1,2T}}(\varphi)(\alpha_{T+1}^1)]|] < C\epsilon
$$
and as $\epsilon>0$ was arbitrary, the term of interest goes to zero in $\mathbb{L}_1$; this completes the proof. The proof can also be repeated if one considers a $K_N(\theta-\theta')$ as part of the function (the argument is almost the same).
The proofs at subsequent times follow the above arguments and are omitted for brevity.
\end{proof}

\section{Proofs for Bias} \label{app:bias}

In the context of the proof for the bias, we need only consider one block (as will become apparent in the proof), as blocks are independent in the asymptotic bias. In addition, one significantly simplify the notations
by simply considering two Feynman-Kac formula of $T$ steps, with different initial distributions, the same potentials and different Markov kernels on measurable spaces $(E_0,\mathcal{E}_0),\dots,(E_T,\mathcal{E}_T)$.
Thus, for $k\in\{1,2\}$ the two Feynman-Kac $n-$time marginals:
$$
\eta_n^k(dx_n) = \frac{\gamma_n^k(dx_n)}{\gamma_n^k(1)}
$$
with 
$$
\gamma_n^k(dx_n) = \int_{\mathbb{R}^{d(n-1)}} \Big\{\prod_{p=0}^{n-1} G_p(x_p) M_{p+1}^k(x_p,dx_{p+1}) \Big\}\eta_0^k(dx_0).
$$
This corresponds to our case, as the potentials are the same, with the Markov kernels and initial distributions different.
Our proofs will depend a lot on the Bayes rule, which we now recall, for $\mu\in\mathscr{P}(E_{p-1})$ 
$$
\Phi_p^k(\mu)(dx) = \frac{\mu(G_{p-1}M_p^k(dx))}{\mu(G_{p-1})}.
$$
We use the notation $\Phi_{p,q}^k(\mu) := \Phi_q^k\circ  \Phi_q^k\circ\cdots \circ\Phi_{p+1}^k(\mu)$, $q\geq p\geq 0$ (with the convention when $p=q$, one returns $\mu$). Our assumption (A\ref{hyp:bias}) under the modified notation is

\begin{itemize}
\item{There exist a $\delta\in[1,\infty)$ such that for every $p\geq 0$
$$
\sup_{x,y\in E_p}\frac{G_p(x)}{G_p(y)} \leq \delta.
$$
}
\item{There exist a $\epsilon\in(0,1)$ and for $T-1\geq p\geq 1$, $\nu\in\mathscr{P}(E_p)$ such that for each $k\in\{1,2\},T-1\geq p\geq 1$ every $x,y\in E_{p-1}$
$$
M_p^k(x,\cdot) \geq \epsilon \Big(M_p^k(y,\cdot)\vee\nu(\cdot)\Big).
$$
}
\end{itemize}
Recall that the final Markov kernel is a Dirac measure.

\begin{proof}[Proof of Theorem \ref{theo:bias}]
We have that, under our modified notations
$$
\mathsf{B}(T,\varphi) = |\Phi_T^1(\eta_0^1)(\varphi)-\Phi_T^2(\eta_0^2)(\varphi)| \leq |\Phi_T^1(\eta_0^1)(\varphi)-\Phi_T^1(\eta_0^1)(\varphi)| + |\Phi_T^1(\eta_0^2)(\varphi)-\Phi_T^2(\eta_0^2)(\varphi)|
$$
By Lemma \ref{lem:forget}
$$
|\Phi_T^1(\eta_0^1)(\varphi)-\Phi_T^1(\eta_0^1)(\varphi)|
\leq \|\varphi\|_{\infty} 4\Big(\frac{\delta}{\epsilon}\Big)^2 (1-\epsilon^2)^{T-1}.
$$
Then
$$
|\Phi_T^1(\eta_0^2)(\varphi)-\Phi_T^2(\eta_0^2)(\varphi)| = |\Phi_{T}^2(\Phi_{T-1}^1(\eta_0^2))(\varphi) -  \Phi_{T}^2(\Phi_{T-1}^2(\eta_0^2))(\varphi)|.
$$
as $\Phi_{T}^1(\mu)(\varphi)=\Phi_{T}^2(\mu)(\varphi)$ for any $\mu\in\mathscr{P}(E_{T-1})$.  Now
\begin{eqnarray*}
\Phi_{T}^2(\Phi_{T-1}^1(\eta_0^2))(\varphi) -  \Phi_{T}^2(\Phi_{T-1}^2(\eta_0^2))(\varphi)  & = &
\frac{[\Phi_{T-1}^1(\eta_0^2)-\Phi_{T-1}^2(\eta_0^2)](G_{T-1}\varphi)}{\Phi_{T-1}^1(\eta_0^2)(G_{T-1})} +\\ & & \frac{\Phi_{T-1}^2(\eta_0^2)(G_{T-1}\varphi)}{\Phi_{T-1}^1(\eta_0^2)(G_{T-1})\Phi_{T-1}^2(\eta_0^2)(G_{T-1})}\times \\ & &
[\Phi_{T-1}^2(\eta_0^2)(G_{T-1})-\Phi_{T-1}^1(\eta_0^2)(G_{T-1})].
\end{eqnarray*}
Then, by application of Proposition \ref{prop:bias} along with (A\ref{hyp:bias}) it follows that
$$
|\Phi_{T}^2(\Phi_{T-1}^1(\eta_0^2))(\varphi) -  \Phi_{T}^2(\Phi_{T-1}^2(\eta_0^2))(\varphi)| \leq 
$$
$$
4\delta^2\|\varphi\|_{\infty}\Big(\frac{\delta^2}{\epsilon^4}\Big) \Big(1-(1-\epsilon^2)^{T-1}\Big) +
4\delta^4\|\varphi\|_{\infty}\Big(\frac{\delta^2}{\epsilon^4}\Big) \Big(1-(1-\epsilon^2)^{T-1}\Big)
$$
which allows one to conclude the proof.
\end{proof}

\begin{rem}
As one can see from inspection of the proof, the difference in initial distribution does not impact the bound. Moreover, the difference in Markov kernels is controlled, leading to a control of the bias; such a latter property is not obvious \emph{a priori} and
needs to be proved. As is evident from the proof, it does not matter which block one considers, under our assumptions.
\end{rem}

\begin{prop}\label{prop:bias}
Assume (A\ref{hyp:bias}). Then for any $\mu\in\mathscr{P}(E_0)$
$$
\|\Phi^1_{T-1}(\mu)-\Phi^2_{T-1}(\mu)\|_{\textrm{\emph{tv}}} \leq 2 \Big(\frac{\delta^2}{\epsilon^4}\Big) \Big(1-(1-\epsilon^2)^{T-1}\Big)
$$
where $\epsilon,\delta$ are as in (A\ref{hyp:bias}).
\end{prop}

\begin{proof}
We have the standard telescoping sum, for $\varphi:E_{T-1}\rightarrow[0,1]$
$$
[\Phi^1_{T-1}(\mu)-\Phi^2_{T-1}(\mu)](\varphi) 
$$
$$
= \sum_{k=0}^{T-1}\Big [\Phi^1_{k,T-1}\Big(\Phi_{k}^2(\mu)\Big) - \Phi^1_{k+1,T-1}\Big(\Phi_{k+1}^2(\mu)\Big)\Big](\varphi).
$$
Now, by Lemma \ref{lem:forget}
$$
\Big\|\Phi^1_{k,T-1}\Big(\Phi_{k}^2(\mu)\Big) - \Phi^1_{k+1,T-1}\Big(\Phi_{k+1}^2(\mu)\Big)\Big\|_{\textrm{tv}} \leq 
$$
$$
2 \Big(\frac{\delta}{\epsilon}\Big)^2 (1-\epsilon^2)^{T-k-2}  \Big\|\Phi^1_{k+1}\Big(\Phi_{k}^2(\mu)\Big) - \Phi^2_{k+1}\Big(\Phi_{k}^2(\mu)\Big) \Big\|_{\textrm{tv}}.
$$
By Lemma \ref{lem:tech_res_bias}
$$
\Big\|\Phi^1_{k+1}\Big(\Phi_{k}^2(\mu)\Big) - \Phi^2_{k+1}\Big(\Phi_{k}^2(\mu)\Big) \Big\|_{\textrm{tv}} \leq (1-\epsilon)
$$
so we have proved that
$$
\|\Phi^1_{pl,(p+1)l}(\mu)-\Phi^2_{pl,(p+1)l}(\mu)\|_{\textrm{tv}} \leq 2 \Big(\frac{\delta}{\epsilon}\Big)^2 (1-\epsilon) \sum_{k=0}^{T-1} (1-\epsilon^2)^{T-k-2},
$$
from which one can easily conclude.
\end{proof}

\begin{lem}\label{lem:tech_res_bias}
Assume (A\ref{hyp:bias}). Then for any $T-1\geq p\geq 1$, $\mu\in\mathscr{P}(E_{p-1})$ and $\varphi:\mathbb{R}^d\rightarrow[0,1]$ we have
$$
|[\Phi_{p}^1(\mu)-\Phi_{p}^2(\mu)](\varphi)| \leq (1-\epsilon)
$$
where $\epsilon$ is as in (A\ref{hyp:bias}).
\end{lem}

\begin{proof}
We have
$$
[\Phi_{p}^1(\mu)-\Phi_{p}^2(\mu)](\varphi) = \frac{1}{\mu(G_{p-1})}\mu(G_{p-1}[M_p^1-M_p^2](\varphi)).
$$
Now, by (A\ref{hyp:bias}), for any $x,k$ one can write
$$
M_p^k(x,dy) = (1-\epsilon) R_p^k(x,dy) + \epsilon \nu(dy) 
$$
where $R_p^k(x,dy)$ is a Markov kernel:
$$
R_p^k(x,dy) = \frac{M_p^k(x,dy)-\epsilon\nu(dy)}{(1-\epsilon)}.
$$
Thus
$$
[M_p^1-M_p^2](\varphi)(x) = (1-\epsilon)[R_p^1-R_p^2](\varphi)(x)
$$
from which one easily concludes.
\end{proof}

The following Lemma just collects some results of \cite{delmoral1} into a convenient form for use in the above proofs. We make the following defintions for $0\leq s < t$, $k\in\{1,2\}$:
$Q_s^k(x,dy) = G_{s-1}(x) M_s^k(x,dy)$,
$$
Q_{s,t}^k(x_p,dx_n) = \int_{\mathbb{R}^{d(t-s-2}} \prod_{q=s+1}^t Q_q^k(x_{q-1},dx_q)
$$ 
and $P_{s,t}^k(x,dy) = Q_{s,t}^k(x,dy)/Q_{s,t}^k(1)(x)$. For a bounded and measurable real-valued function $\varphi$ we denote
$$
\|\varphi\|_{\textrm{osc}} = \sup_x |\varphi(x)| + \sup_{x,y}|\varphi(x)-\varphi(y)|.
$$

\begin{lem}\label{lem:forget}
Assume (A\ref{hyp:bias}). Then for any $0\leq s < t\leq T-1$ and any $\mu,\rho\in\mathscr{P}(E_{s-1})$, $k\in\{1,2\}$
$$
\|\Phi_{s,t}^k(\mu)-\Phi_{s,t}^k(\rho)\|_{\textrm{\emph{tv}}} \leq 2\Big(\frac{\delta}{\epsilon}\Big)^2 (1-\epsilon^2)^{t-s} \|\mu-\rho\|_{\textrm{\emph{tv}}}
$$
where $\epsilon,\delta$ are as in (A\ref{hyp:bias}).
\end{lem}

\begin{proof}
By \cite[Theorem 4.3.1]{delmoral1} 
$$
\|\Phi_{s,t}^k(\mu) - \Phi_{s,t}^k(\rho) \|_{\textrm{tv}} \leq \beta(P_{s,t}^k) \frac{\|Q_{s,t}^k(1)\|_{\textrm{osc}}}{\mu(Q_{s,t}^k(1))\wedge \rho(Q_{s,t}^k(1))} \|\mu-\rho\|_{\textrm{tv}}
$$
where  $\beta(P_{s,t}^k)$ is the Dobrushin coefficient of $P_{s,t}^k$. Then by \cite[pp.~142]{delmoral1} (fourth displayed equation) and (A\ref{hyp:bias})
$$
\beta(P_{s,t}^k) \leq (1-\epsilon^2)^{t-s}.
$$
Note that via (A\ref{hyp:bias}) and \cite[Lemma 4.1]{cerou1} one has
$$
\sup_{x,y}\frac{Q_{s,t}^k(1)(x)}{Q_{s,t}^k(1)(y)} \leq \frac{\delta}{\epsilon}
$$
and as a result via \cite[pp.~138]{delmoral1}, for any pair of probabilities $\mu,\rho\in\mathscr{P}(E_{s-1})$
$$
\frac{\|Q_{s,t}^k(1)\|_{\textrm{osc}}}{\mu(Q_{s,t}^k(1))\wedge \rho(Q_{s,t}^k(1))} \leq 2 \Big(\frac{\delta}{\epsilon}\Big)^2 .
$$
The result thus follows.
\end{proof}

\end{document}